\documentclass[a4paper,11pt]{article}
\usepackage[utf8]{inputenc}
\RequirePackage{amsthm}
\RequirePackage{amssymb}
\RequirePackage{amsfonts}
\RequirePackage[tbtags,fleqn]{amsmath}

\usepackage[T1]{fontenc}
\usepackage{lmodern}
\RequirePackage{textcomp}

\RequirePackage[english,UKenglish]{babel}
\RequirePackage{authblk}
\RequirePackage{soul}
\RequirePackage{color}

\RequirePackage{graphicx}
\RequirePackage{hyperref}
\hypersetup{pdfborder={0 0 0}}

\RequirePackage[labelsep=space,singlelinecheck=false,  font={up,small},labelfont={sf,bf},  listof=false]{caption}  

\theoremstyle{plain}
\newtheorem{theorem}{Theorem}

\newtheorem{claim}[theorem]{Claim}
\newtheorem{lemma}[theorem]{Lemma}

\newtheorem{proposition}[theorem]{Proposition}

\theoremstyle{definition}

\theoremstyle{remark}

\usepackage{fullpage}

\RequirePackage[english,UKenglish]{babel}

\usepackage{float}
\usepackage{wrapfig}
\usepackage{collect}

\newtheorem{observation}[theorem]{Observation}

\newcommand{\paren}[1]{\left( #1 \right) }

\newcommand{\opt}{\mbox{\sf OPT}}

\newcommand{\set}[1]{\left\{ #1 \right\}}

\newcommand{\sset}{{\mathcal{S}}}
\newcommand{\xset}{{\mathcal{X}}}
\newcommand{\lset}{{\mathcal{L}}}
\newcommand{\tset}{{\mathcal T}}

\newcommand{\bset}{{\mathcal{B}}}

\newcommand{\eset}{{\mathcal{E}}}

\newcommand{\uset}{{\mathcal{U}}}

\newcommand{\eps}{\epsilon}

\newcommand{\expect}[2]{\text{\bf E}_{#1}\left [#2\right]}

\newcommand{\pr}[2]{\text{\bf Pr}_{#1}\left [#2\right ]}

\def\daniel#1{}
\def\parinya#1{}

\definecollection{sec:prelim}
\definecollection{sec:gap}
\definecollection{sec:newlp}

\newenvironment{appendixproof}[2]{ \proof }{ \endproof }
\newenvironment{appendixextra}[2]{}{}

\begin{document}

\title{New Integrality Gap Results for the Firefighters Problem on Trees}

\author[1]{Parinya Chalermsook}
\author[1,2]{Daniel Vaz}
\affil[1]{Max-Planck-Institut für Informatik\\
  Saarbrücken, Germany}
\affil[2]{Graduate School of Computer Science, Saarland University\\ Saarbrücken, Germany
\texttt{\{parinya, ramosvaz\}@mpi-inf.mpg.de}}

\maketitle

\begin{abstract} 
In the firefighter problem on trees, we are given a tree $G=(V,E)$ together with a vertex $s \in V$ where the fire starts spreading. 
At each time step, the firefighters can pick one vertex while the fire spreads from burning vertices to all their neighbors that have not been picked.
The process stops when the fire can no longer spread. 
The objective is to find a strategy that maximizes the total number of vertices that do not burn.
This is a simple mathematical model, introduced in 1995, that abstracts the spreading nature of, for instance, fire, viruses, and ideas.  
The firefighter problem is NP-hard and admits a $(1-1/e)$ approximation based on rounding the canonical LP. Recently, a PTAS was announced\cite{2016arXiv160100271A}.~\footnote{The $(1-1/e)$ approximation remained the best until very recently when Adjiashvili et al.~\cite{2016arXiv160100271A} showed a PTAS. Their PTAS does not bound the LP gap.}  

The goal of this paper is to develop better understanding on the power of LP relaxations for the firefighter problem. 
We first show a matching lower bound of $(1-1/e+\epsilon)$ on the integrality gap of the canonical LP.
This result relies on a powerful {\em combinatorial gadget} that can be used to derive integrality gap results in other related settings.  
Next, we consider the canonical LP augmented with simple additional constraints (as suggested by Hartke).  
We provide several evidences that these constraints improve the integrality gap of the canonical LP: (i) Extreme points of the new LP are integral for some known tractable instances and (ii) A natural family of instances that are bad for the canonical LP admits an improved approximation algorithm via the new LP.  
We conclude by presenting a $5/6$ integrality gap instance for the new LP.  
\end{abstract} 

\section{Introduction}

Consider the following graph-theoretic model that abstracts the fire spreading process: We are given graph $G=(V,E)$ together with the source vertex $s$ where the fire starts. 
At each time step, we are allowed to pick some vertices in the graph to be saved, and the fire spreads from burning vertices to their neighbors that have not been saved so far.
The process terminates when the fire cannot spread any further.
This model was introduced in 1995~\cite{hartnell1995firefighter} and has been used extensively by researchers in several fields as an abstraction of epidemic propagation. 

There are two important variants of the firefighters problem. 
(i) In the maximization variant ({\sc Max-FF}), we are given graph $G$ and source $s$, and we are allowed to pick one vertex per time step.  
The objective is to maximize the number of vertices that do not burn.
And (ii) In the minimization variant ({\sc Min-FF}), we are given a graph $G$, a source $s$, and a terminal set $\xset \subseteq V(G)$, and we are allowed to pick $b$ vertices per time step. 
The goal is to save all terminals in $\xset$, while minimizing the budget $b$.  

In this paper, we focus on the {\sc Max-FF} problem. 
The problem is $n^{1-\epsilon}$ hard to approximate in general graphs~\cite{AnshelevichCHS12}, so there is no hope to obtain any reasonable approximation guarantee.
Past research, however, has focused on sparse graphs such as trees or grids. 
Much better approximation algorithms are known on trees: The problem is NP-hard~\cite{king2010firefighter} even on trees of degree at most three, but it admits a $(1-1/e)$ approximation algorithm. 
For more than a decade~\cite{AnshelevichCHS12,ChalermsookC10,CaiVY08,develin2007fire,IwaikawaKM11,king2010firefighter}, there was no progress on this approximability status of this problem, until a PTAS was recently discovered \cite{2016arXiv160100271A}.

Besides the motivation of studying epidemic propagation, the firefighter problem and its variants are interesting due to their connections to other classical optimization problems: 

\begin{itemize} 
\item (Set cover) 
The firefighter problem is a special case of the {\em maximum coverage problem with group budget constraint (MCG)}~\cite{ChekuriK04}: 
Given a collection of sets $\sset= \{S_1,\ldots, S_m\}: S_i \subseteq X$, together with {\em group constraints}, i.e. a partition of $\sset$ into groups $G_1,\ldots, G_{\ell}$, we are interested in choosing one set from each group in a way that maximizes the total number of elements covered, i.e. a feasible solution is a subset $\sset' \subseteq \sset$ where $|\sset' \cap G_j| \leq 1$ for all $j$, and $|\bigcup_{S_i \in \sset'} S_i|$ is maximized.
It is not hard to see that {\sc Max-FF} is a special case of MCG. 
We refer the readers to the discussion by Chekuri and Kumar~\cite{ChekuriK04} for more applications of MCG.

\item (Cut) In a standard minimum node-cut problem, we are given a graph $G$ together with a source-sink pair $s,t \in V(G)$. 
Our goal is to find a collection of nodes $V' \subseteq V(G)$ such that $G \setminus V'$ has $s$ and $t$ in distinct connected components.  
Anshelevich et al.~\cite{AnshelevichCHS12} discussed that the firefighters' solution can be seen as a ``cut-over-time'' in which the cut must be produced gradually over many timesteps. 
That is, in each time step $t$, the algorithm is allowed to choose vertex set $V'_t$ to remove from the graph $G$, and again the final goal is to ``disconnect'' $s$ from $t$~\footnote{The notion of disconnecting the vertices here is slightly non-standard. }.  
This cut-over-time problem is exactly equivalent to the minimization variant of the firefighter problem. 
We refer to~\cite{AnshelevichCHS12} for more details about this equivalence. 
 
\end{itemize}

\subsection{Our Contributions }
In this paper, we are interested in developing a better understanding of the {\sc Max-FF} problem from the perspective of LP relaxation. 
The canonical LP relaxation has been used to obtain the\daniel{} known $(1-1/e)$ approximation algorithm via straightforward independent LP rounding (each node is picked independently with the probability proportional to its LP-value). 
So far, it was not clear whether an improvement was possible via this LP, for instance, via sophisticated dependent rounding schemes~\footnote{Cai, Verbin, and Yang~\cite{CaiVY08} claimed an LP-respecting integrality gap of $(1-1/e)$, but many natural rounding algorithms in the context of this problem are not LP respecting, e.g. in~\cite{ChalermsookC10}}.  
Indeed, for the corresponding minimization variant, {\sc Min-FF}, Chalermsook and Chuzhoy designed a dependent rounding scheme for the canonical LP in order to obtain $O(\log^* n)$ approximation algorithm, improving upon an $O(\log n)$ approximation obtained via independent LP rounding.   
In this paper, we are interested in studying this potential improvement for {\sc Max-FF}.

Our first result refutes such possibility for {\sc Max-FF}.  
In particular, we show that the integrality gap of the standard LP relaxation can be arbitrarily close to $(1-1/e)$.

\begin{theorem}
\label{thm: main}  
For any $\epsilon >0$, there is an instance $(G,s)$ (whose size depends on $\epsilon$) such that the ratio between optimal integral solution and fractional one is at most $(1- 1/e+ \epsilon)$.  
\end{theorem} 

Our techniques rely on a powerful {\em combinatorial gadget} that can be used to prove integrality gap results in some other settings studied in the literature.  
In particular, in the $b$-{\sc Max-FF} problem, the firefighters can pick up to $b$ vertices per time step, and the goal is to maximize the number of saved vertices. 
We provide an integrality gap of $(1-1/e)$ for the {\sc $b$-Max-FF} problem for all constant $b \in {\mathbb N}$, thus matching the algorithmic result of~\cite{CostaDDPR13}.
In the setting where an input tree has degree at most $d \in [4,\infty)$, we show an integrality gap result of $(1-1/e + O(1/\sqrt{d}))$.  
The best known algorithmic result in this setting was previously a $(1-1/e+ \Omega(1/d))$ approximation due to~\cite{IwaikawaKM11}.  \daniel{}

Motivated by the aforementioned negative results, we search for a stronger LP relaxation for the problem. 
We consider adding a set of valid linear inequalities, as suggested by Hartke~\cite{Hartke2006}. 
We show the following evidences that the new LP is a stronger relaxation than the canonical LP.  
\begin{itemize} 
\item Any extreme point of the new LP is integral for the tractable instances studied by Finbow et al.~\cite{FinbowG09}. 
In contrast, we argue that the canonical LP does not satisfy this integrality property of extreme points.  

\item A family of instances which captures the bad integrality gap instances given in Theorem~\ref{thm: main}, admits a better than $(1-1/e)$ approximation algorithm via the new LP. 

\item When the LP solution is near-integral, e.g. for half-integral solutions, the new LP is provably better than the old one. 
 
\end{itemize}  

Our results are the first rigorous evidences that Hartke's constraints lead to improvements upon the canonical LP. 
All the aforementioned algorithmic results exploit the new LP constraints in dependent LP rounding procedures.  
In particular, we propose a two-phase dependent rounding algorithm, which can be used in deriving the second and third results.  
We believe the new LP has an integrality gap strictly better than $(1-1/e)$, but we are unable to analyze it.  

Finally, we show a limitation of the new LP by presenting a family of instances, whose integrality gap can be arbitrarily close to $5/6$. 
This improves the known integrality gap ratio~\cite{Hartke2006}, and puts the integrality gap answer somewhere between $(1-1/e)$ and $5/6$. 
Closing this gap is, in our opinion, an interesting open question.  

{\bf Organization:} In Section~\ref{sec:prelim}, we formally define the problem and present the LP relaxation. 
In Section~\ref{sec:gap}, we present the bad integrality gap instances. 
We present the LP augmented with Hartke's constraints in Section~\ref{sec:newlp} and discuss the relevant evidences of its power in comparison to the canonical LP. Some proofs are omitted for space constraint, and are presented in Appendix.

{\bf Related results:}
King and MacGillivray showed that the firefighter problem on trees is solvable in polynomial time if the input tree has degree at most three, with the fire starting at a degree-2 vertex. 
From exponential time algorithm's perspective, Cai et al. showed $2^{O(\sqrt{n} \log n)}$ time, exact algorithm.  
The discrete mathematics community pays particularly high attention to the firefighter problem on grids~\cite{wang2002fire,develin2007fire}, and there has also been some work on infinite graphs~\cite{hartnell1995firefighter}.   
 
The problem also received a lot of attention from the parameterized complexity perspectives~\cite{ChlebikovaC14,BazganCCFFL14,CaiVY08} and on many special cases, e.g., when the tree has bounded pathwidth~\cite{ChlebikovaC14} and on bounded degree graphs~\cite{ChlebikovaC14,BazganCR13}.

{\bf Recent update:} Very recently, Adjiashvili et al.~\cite{2016arXiv160100271A} showed a polynomial time approximation scheme (PTAS) for the {\sc Max-FF} problem, therefore settling the approximability status.
Their results, however, do not bound the LP integrality gap.  
We believe that the integrality gap questions are interesting despite the known approximation guarantees.

\section{Preliminaries} 
\label{sec:prelim} 

A formal definition of the problem is as follows. 
We are given a graph $G$ and a source vertex $s$ where the fire starts spreading. 
A {\em strategy} is described by a collection of vertices $\uset= \{u_{t}\}_{t=1}^n$ where $u_t \in V(G)$ is the vertex picked by firefighters at time $t$.  
We say that a vertex $u \in V(G)$ is {\em saved} by the strategy $\uset$ if for all path $P=(s=v_0,\ldots, v_z =u)$ from $s$ to $u$, we have $v_i \in \{u_1,\ldots, u_{i}\}$ for some $i=1,\ldots, z$. 
A vertex $v$ not saved by $\uset$ is said to be a {\em burning vertex}.  
The objective of the problem is to compute $\uset$ so as to maximize the total number of saved vertices. 
Denote by $\opt(G,s)$ the number of vertices saved by an optimal solution.  

When $G$ is a tree, we think of $G$ as being partitioned into layers $L_1,\ldots, L_{\lambda}$ where $\lambda$ is the height of the tree, and $L_i$ contains vertices whose distance is exactly $i$ from $s$. 
Every strategy has the following structure.  
 
\begin{proposition}
Consider the firefighters problem's instance $(G,s)$ where $G$ is a tree.  
Let $\uset=\{u_1,\ldots, u_n\}$ be any strategy. 
Then there is another strategy $\uset'=\{u'_t\}$ where $u'_t$ belongs to layer $t$ in $G$, and $\uset'$ saves at least as many vertices as $\uset$ does.  
\end{proposition}  

We remark that this structural result only holds when an input graph $G$ is a tree. 

\vspace{0.1in}

{\bf LP Relaxation:} 
This paper focuses on the linear programming aspect of the problem.
For any vertex $v$, let $P_v$ denote the (unique) path from $s$ to $v$, and let $T_v$ denote the subtree rooted at $v$.  
A natural LP relaxation is denoted by (LP-1): We have variable $x_v$ indicating whether $v$ is picked by the solution, and $y_v$ indicating whether $v$ is saved.  

\vspace{0.1in}

\noindent
\framebox[\textwidth]{ 
\begin{minipage}[b]{0.48\textwidth}
\vspace{-1em}
\begin{align*} 
 \mbox{(LP-1)} &  \\
  \max & \sum_{v \in V} y_v\\
  & \sum_{v \in L_j} x_v \leq 1 \mbox{ for all layer $j$} \\ 
  &  y_v \leq \sum_{u \in P_v} x_u \mbox{ for all $v\in V$}\\
  & x_v, y_v \in [0,1] \mbox{ for all $v$ }  
\end{align*}
\end{minipage}\hfill
\begin{minipage}[b]{0.48\textwidth}
\vspace{-1em}
\begin{align*} 
 \mbox{(LP-2)}&  \\
 \max & \sum_{v \in \xset} y_v\\
 & \sum_{v \in L_j} x_v \leq 1 \mbox{ for all layer $j$} \\ 
 &  y_v \leq \sum_{u \in P_v} x_u \mbox{ for all $v\in \xset$}\\
 & x_v, y_v \in [0,1] \mbox{ for all $v$ }  
\end{align*} 
\end{minipage}} 

\vspace{0.1in}

Let ${\sf LP}(T,s)$ denote the optimal fractional LP value for an instance $(T,s)$. 
The integrality gap ${\sf gap}(T,s)$ of the instance $(T,s)$ is defined as ${\sf gap}(T,s) = {\sf OPT}(T,s)/{\sf LP}(T,s)$. 
The integrality gap of the LP is defined as $\inf_T {\sf gap}(T, s)$.

\vspace{0.1in}

{\bf Firefighters with terminals:}
We consider a more general
variant of the problem, where we are only interested in saving a subset
$\xset$ of vertices, which we call {\em terminals}. 
The goal is now to maximize the number of saved terminals. 
An LP formulation of this problem, given an instance $(T,v, \xset)$, is denoted by (LP-2).  
The following lemma argues that these two variants are ``equivalent'' from the perspectives of LP relaxation. 

\begin{lemma} 
Let $(T, \xset, s)$, with $|\xset| > 0$, be an input for the terminal firefighters problem that gives an integrality gap of $\gamma$ for (LP-2), and that the value of the fractional optimal solution is at least $1$. 
Then, for any $\epsilon >0$, there is an instance $(T',s')$ that gives an integrality gap of $\gamma + \epsilon$ for (LP-1). 
\label{lem:reduction}
\end{lemma}

\begin{appendixproof}{sec:prelim}{\subsection{Proof of Lemma \ref{lem:reduction}}}
Let $M=2 |V(T)|/\epsilon$. 
Starting from $(T,\xset, s)$, we construct an instance $(T',s')$ by adding $M$ children to each vertex in $\xset$,
so the number of vertices in $T'$ is $|V(T')| = |V(T)| + M|\xset|$.  
We denote the copies of $\xset$ in $T'$ by $\xset^\prime$ and the set of their added children by $\xset''$.  
The root of the new tree, $s^\prime$, is the same as $s$ (the root of $T$.)
Now we argue that the instance $(T',s')$ has the desired integrality gap, i.e. we argue that $\opt(T',s') \leq (\gamma + \epsilon) {\sf LP} (T',s')$.  

Let $(x^\prime, y^\prime)$ be an integral solution to the instance
$(T',s')$. 
We upper bound the number of vertices saved by this solution, i.e. upper bounding $\sum_{v \in V(T')} y'_v$. 
We analyze three cases: 
\begin{itemize} 
\item For a vertex $v \in V(T') \setminus \xset''$, we upper bound the term $y'_v$ by $1$, and so the sum $\sum_{v \in V(T') \setminus \xset''} y'_v$ by $|V(T)|$. 

\item Now define $\tilde \xset \subseteq \xset''$ as the set of vertices $v$ for which $y'_v =1$ but $y'_{u}= 0$ for the parent $u$ of $v$.
This means that $x'_v = 1$ for all vertices in $\tilde \xset$.  
Notice that $\sum_{v \in \tilde \xset} y'_v \leq |\xset|$: We break the set $\tilde \xset$ into $\set{\tilde \xset_u}_{u \in V(T')}$ where $\tilde \xset_u =\set{v \in \tilde \xset: \mbox{ $u$ is the parent of $v$}}$.
The LP constraint guarantees that $\sum_{v \in \tilde \xset_u} y'_v = \sum_{v \in \tilde \xset_u} x'_v \leq 1$ (all vertices in $\tilde \xset_u$ belong to the same layer.)  
Summing over all such $u \in \xset'$, we get the desired bound.

\item Finally, consider the term $\sum_{v \in \xset' \setminus \tilde \xset}y'_v$. Let $(x^*, y^*)$ be an optimal fractional solution to $(T, \xset,s)$ for (LP-2).  
We only need to care about all vertices $v$ such that $y'_u =1$ for the parent $u$ of $v$.  
This term is upper bounded by $M \sum_{u \in \xset'} y'_u$, which is at most $M \gamma \paren{\sum_{v \in \xset} y^*_v}$, due to the fact that the solution $(x',y')$ induces an integral solution on instance $(T,\xset, s)$.  

\end{itemize}

Combining the three cases, we get $|V(T)| + |\xset| + M \gamma \paren{\sum_{v \in \xset} y^*_v} \leq M + \gamma M \paren{\sum_{v \in \xset} y^*_v} \leq M (\gamma+ \epsilon) \paren{\sum_{v \in \xset} y^*_v}$, if $\sum_{v \in \xset} y^*_v \geq 1$.
Now, notice that the fractional solution $(x^*,y^*)$ for (LP-2) on instance $(T,\xset, s)$ is also feasible for (LP-1) on $(T',s')$ with at least a multiplicative factor of $M$ times the objective value: Each fractional saving of $u \in \xset'$ contributes to save all $M$ children of $u$.  
Therefore, $M\paren{ \sum_{v \in \xset} y^*_v} \leq M \cdot {\sf LP}(T',s')$, thus concluding the proof.     
\end{appendixproof}

We will, from now on, focus on studying the integrality gap of (LP-2).

\section{Integrality Gap of (LP-2)}

\label{sec:gap} 

We first discuss the integrality gap of (LP-2) for a general tree.
We use the following combinatorial gadget.  

\vspace{0.1in} 

{\bf Gadget:} 
A $(M, k, \delta)$-good gadget is a collection of trees $\tset = \{T_1,\ldots, T_M\}$, with roots $r_1, \ldots, r_M$ where $r_i$ is a root of $T_i$, and a subset $\sset \subseteq \bigcup V(T_i)$ that satisfy the following properties:

\begin{itemize}
 
\item (Uniform depth) We think of these trees as having layers $L_0,L_1,\ldots, L_h$, where $L_j$ is the union over all trees of all vertices at layer $j$ and $L_0 = \{r_1, \ldots, r_m\}$. 
All leaves are in the same layer $L_h$.

\item (LP-friendly) 
For any layer $L_j$, $j\geq 1$, we have $|\sset \cap L_j| \leq k$. 
Moreover, for any tree $T_i$ and a leaf $v \in V(T_i)$, the unique path from $r_i$ to $v$ must contain at least one vertex in $\sset$.  
 
\item (Integrally adversarial) Let $\bset \subseteq \set{r_1,\ldots, r_M}$ be any subset of roots. 
Consider a subset of vertices $\uset = \{u_j\}_{j=1}^h$ such that $u_j \in L_j$. 
 For $r_i \in \bset$ and a leaf $v \in L_h\cap V(T_i)$, we say that $v$ is {\em $(\uset, \bset)$-risky} if the unique path from $r_i$ to $v$ does not contain any vertex in $\uset$.
There must be at least $(1-1/k - \delta) \frac{|\bset|}{M} |L_h|$ vertices in $L_h$ that are $(\uset,\bset)$-risky, for all choices of $\bset$ and $\uset$.  
\end{itemize} 

We say that vertices in $\sset$ are {\em special} and all other vertices are {\em regular}.

\begin{lemma} 
\label{gap:lemma:goodgadget}
For any integers $k\geq2$, $M \geq 1$, and any real number $\delta >0$, a $(M,k ,\delta)$-good gadget exists. Moreover, the gadget contains at most $ (k/\delta)^{O(M)}$ vertices.  
\end{lemma} 

We first show how to use this lemma to derive our final construction. 
The proof of the lemma follows later. 

\vspace{0.1in} 

{\bf Construction:}
Our construction proceeds in $k$ phases, and we will define it inductively. 
The first phase of the construction is simply a $(1,k,\delta)$-good gadget. 
Now, assume that we have constructed the instance up to phase $q$. 
Let
$l_1,\ldots, l_{M_q} \in L_{\alpha_p}$ be the leaves after the construction of
phase $q$ that all lie in layer $\alpha_q$. 
In phase $q+1$, we take the $(M_q, k, \delta)$-good gadget $(\tset_q, \{r_q\},\sset_q)$; recall that such a gadget consists of $M_q$ trees.   
For each $i =1,\ldots, M_q$, we unify each root $r_i$ with the leaf $l_i$.  
This completes the description of the construction.

Denote by $\bar{\sset}_q= \bigcup_{q' \leq q} \sset_{q'}$ the set of all special vertices in the first $q$ phases. 
After phase $q$, we argue that our construction
satisfies the following properties:

\begin{itemize}
 
\item All leaves are in the same layer $\alpha_q$.  

\item For every layer $L_j$, $|L_j \cap \bar{\sset}_q| \leq k$.  For every path $P$ from the root to $v \in L_{\alpha_i}$, $|P \cap \bar{\sset}_q| = q$.

\item For any integral solution $\uset$, at least
$|L_{\alpha_q}| \left(\left(1-1/k\right)^q-q\delta\right)$ vertices of $L_{\alpha_q}$ burn.

\end{itemize}

It is clear from the construction that the leaves after phase $q$ are all in
the same layer. 
As to the second property, the properties of the gadget ensure
that there are at most $k$ special vertices per layer. 
Moreover, consider each path $P$ from the root to some vertex $v \in L_{\alpha_{q+1}}$.
We can split this path into two parts $P = P' \cup P''$ where $P'$ starts from the root and ends at some $v' \in L_{\alpha_{q}}$, and $P''$ starts at $v'$ and ends at $v$.   
By the induction hypothesis, $|P'\cap \bar{S}_{q}| = q$ and the second property of the gadget guarantees that $|P'' \cap \sset_{q+1}| = 1$.

To prove the final property, consider a solution $\uset=\{u_1,\ldots, u_{\alpha_{q+1}}\}$, which can be seen as $\uset'\,\cup\,\uset''$ where $\uset' = \{u_1,\ldots, u_{\alpha_q}\}$ and $\uset'' = \{u_{\alpha_{q}+1}, \ldots, u_{\alpha_{q+1}}\}$. 
By the induction hypothesis, we have that at least  
$\left((1-1/k)^{q} - q \delta\right) |L_{\alpha_q}|$ vertices in $L_{\alpha_q}$ burn; denote these burning vertices by $\bset$. 
The third property of the gadget will ensure that at least $(1-1/k- \delta)\frac{|\bset|}{M_{q}} |L_{\alpha_{q+1}}|$ vertices in $L_{\alpha_{q+1}}$ must be $(\uset'', \bset)$-risky.  
For each risky vertex $v \in L_{\alpha_{q+1}}$, a unique path from the root to $v' \in \bset$ does not contain any vertex in $\uset'$, and also the path from $v'$ to $v$ does not contain a vertex in $\uset''$ (due to the fact that it is $(\uset'', \bset)$-risky.)    
This implies that such vertex $v$ must burn. 
Therefore, the fraction of burning vertices in layer $L_{\alpha_{q+1}}$ is at least $(1-1/k-\delta)|\bset|/M_q \geq (1-1/k - \delta)((1-1/k)^q- q\delta)$, by induction hypothesis.  
This number is at least $(1-1/k)^{q+1} - (q+1) \delta$, maintaining the invariant. 

After the construction of all $k$ phases, the leaves are
designated as the terminals $\xset$.
Also, $M_{q+1} \leq (k/\delta)^{2M_q}$, which means that, after $k$
phases,  $M_k$ is at most a tower function of $(k/\delta)^2$, that is,
$(k/\delta)^{2(k/\delta)^{\cdots}}$ with $k-1$ such exponentiations.
The total size of the construction is $\sum_q  (k/\delta)^{2M_q} \leq
(k/\delta)^{2M_k} = O(M_{k+1})$.

An example construction, when $k=2$, is presented in Figure \ref{11eintgap:gadget} (in Appendix). 

\begin{theorem}
\label{thm: gap 1/k}  
A fractional solution, that assigns $x_v = 1/k$ to each special vertex $v$, saves every terminal. On the other hand, any integral solution can save at most a fraction of $1-(1-1/k)^k + \epsilon$.  
\end{theorem}   
\begin{appendixproof}{sec:gap}{\subsection{Proof of Theorem \ref{thm: gap 1/k}}}
We assign the LP solution $x_v = 1/k$ to all special vertices (those vertices in $\bar{\sset}_k$), and $x_v = 0$
to regular vertices.  
Since the construction ensures that there are at
most $k$ special vertices per layer, we have $\sum_{v \in L_j} x_v \leq
1$ for every layer $L_j$. 
Moreover, every terminal is fractionally saved: For any $t \in \xset$, the path $|P_t \cap \bar{\sset}_k| = k$, so we have $\sum_{v \in P_t} x_v = 1$.  

For the integral solution analysis, set $\delta = \epsilon / k$. The proof follows immediately from the properties of the instance.
\end{appendixproof} 

\subsection{Proof of Lemma~\ref{gap:lemma:goodgadget}}
\label{sec:good-gadget} 
We now show that the $(M, k,\delta)$-good gadget exists for any value of $M \in {\mathbb N}$, $k \in {\mathbb N}, k\geq 2$ and $\delta \in {\mathbb R}_{>0}$. 
We first describe the construction and then show that it has the desired properties.

\vspace{0.1in} 

{\bf Construction:} 
Throughout the construction, we use a structure which we call \emph{spider}. A
spider is a tree in which every node except the root has at most one child. If
a node has no children (i.~e.~a leaf), we call it a \emph{foot} of the spider.
We call the paths from the root to each foot the \emph{legs} of the spider.

Let $D=\lceil4/\delta\rceil$. For each $i = 1,\ldots, M$, the tree $T_i$ is constructed as follows.  
We have a spider  rooted at $r_i$ that contains $kD^{i-1}$ legs.
Its feet are in $D^{i-1}$ consecutive layers, starting at layer $\alpha_i = 1
+ \sum_{j<i}D^{j-1}$; each such layer has $k$ feet. 
Denote by $\sset^{(i)}$ the feet of these spiders. 
Next, for each vertex $v \in \sset^{(i)}$, we have a spider rooted at $v$, having $D^{2M-i+1}$ feet, all of which belong to layer $\alpha=1
+ \sum_{j\leq M}D^{j-1}$.
The set $\sset$ is defined as $\sset = \bigcup_{i=1}^M \sset^{(i)}$. 
This concludes the construction.
We will use the following observation:

\begin{observation} 
For each root $r_i$, the number of leaves of $T_i$ is $k D^{2M}$.
\label{obs: same size}  
\end{observation} 

{\bf Analysis:} 
We now prove that the above gadget is $(M, k,\delta)$-good.
The construction ensures that all leaves are in the same layer $L_{\alpha}$.  

The second property also follows obviously from the construction: For $i \neq i'$, we have that $\sset^{(i)} \cap \sset^{(i')} = \emptyset$, and that each layer contains exactly $k$ vertices from $\sset^{(i)}$.  
Moreover, any path from $r_i$ to the leaf of $T_i$ must go through a vertex in $\sset^{(i)}$.

The third and final property is established by the following two lemmas.

\begin{lemma} 
For any $r_i \in \bset$ and any subset of vertices $\uset = \{u_j\}_{j=1}^h$ 
such that $u_j \in L_j$, a fraction of at least $(1-1/k-2/D)$ of $\sset^{(i)}$ are $(\uset,\bset)$-risky.
\label{lemma:phase1:prop}
\end{lemma} 

\begin{proof}Notice that a vertex $v$ is $(\uset,\bset)$-risky if $\uset$ is not a vertex cut separating $v$ from $\bset$. 
There are $k D^{i-1}$ vertex-disjoint paths from $r_i \in \bset$ to vertices in $\sset^{(i)}$.  
But the cut $\uset$ induced on these paths contains at most $\sum_{i' \leq i} D^{i'-1}$ vertices (because all vertices in $\sset^{(i)}$ are contained in the first $\sum_{i' \leq i} D^{i'-1} \leq D^{i-1}+ 2D^{i-2}$ layers.) 
Therefore, at most $(1/k + 2/D)$ fraction of vertices in $\sset^{(i)}$ can be disconnected by $\uset$, and those that are not cut remain $(\uset, \bset)$-risky.  
\end{proof} 

\begin{lemma} 
Let $v \in \sset^{(i)}$ that is $(\uset, \bset)$-risky. 
Then at least $(1-2/D)$ fraction of descendants of $v$ in $L_{\alpha}$ must be $(\uset, \bset)$-risky.  
\label{lemma:phase2:prop}
\end{lemma} 

\begin{proof}Consider each $v \in \sset^{(i)}$ that is $(\uset, \bset)$-risky and a collection of leaves $\lset_v$ that are descendants of vertex $v$.  
Notice that a leaf $u \in \lset_v$ is $(\uset,\bset)$-risky if removing $\uset$ does not disconnect vertex $v$ from $u$. 

There are $D^{2M-i+1} \geq D^{M+1}$ vertex disjoint paths connecting vertex $v$ with leaves in $\lset_v$, while the cut set $\uset$ contains at most $2D^M$ vertices. 
Therefore, removing $\uset$ can disconnect at most $2/D$ fraction of vertices in $\lset_v$ from $v$.
\end{proof}

Combining the above two lemmas, for each $r_i \in \bset$, the fraction of leaves of $T_i$ that are $(\uset, \bset)$-risky are at least $(1- 1/k- 2/D)(1-2/D) \geq (1-1/k - 4/D)$.  
Therefore, the total number of such leaves, over all trees in $\tset$, are $(1- 1/k - \delta)|\bset||L_{\alpha}|/M$. 

We extend the cosntruction to other settings in Section \ref{sec:ext1} (in Appendix).

\begin{appendixextra}{sec:gap}{\subsection{Extensions}}
\label{sec:ext1} 

{\bf Arbitrary number of firefighters:} 
Let $b \in {\mathbb N}$. 
In the $b$-firefighter problem, at each time step, the firefighters may choose up to $b$ vertices, and the fire spreads from the burning vertices to vertices that have not been chosen so far.
The goal is to maximize the number of saved vertices.  
In this section, we show the following: 

\begin{theorem} 
For any integer $b \in {\mathbb N}$ (independent of $|V(G)|$),  the integrality gap of the canonical LP can be arbitrarily close to $(1-1/e)$. 
\label{thm:intgap:bff}
\end{theorem} 

\begin{proof}
To prove this theorem, one only need to design a slightly different good gadget. 
That is, an $(M,k, \delta)$-good gadget is now a collection of trees $\tset$ with roots $r_1,\ldots, r_M$ together with $\sset \subseteq \bigcup V(T_i)$  that satisfy the following properties: 

\begin{itemize} 
\item All leaves of $T_i$ are in the same layer $L_h$. 

\item For each layer $L_j$, we have $|\sset \cap L_j| \leq kb$. Moreover, for any tree $T_i$ and a leaf $v \in V(T_i)$, the unique path from $r_i$ to $v$ must contain at least one vertex in $\sset$. 

\item For any subset $\bset \subseteq \set{r_1,\ldots, r_M}$ of roots and for any strategy $\uset$, at least $(1-1/k-\delta)\frac{|\bset| |L_h|}{M}$ vertices in $L_h$ are $(\uset, \bset)$-risky.   
\end{itemize} 

It is not hard to see that these gadgets can be used to construct the integrality gap in the same way as in the previous section. 
Details are omitted. 
\end{proof}

\noindent {\bf Bounded degrees:} 
Iwakawa et al. showed a $(1-1/e + \Omega(1/d))$ approximation algorithm for the instance that has degree at most $d$.
We show an instance where this dependence on $1/d$ is almost the best possible that can be obtained from this LP.  

\begin{theorem}
For all $d \geq 4$, 
the integrality gap of (LP-1) on degree-$d$ graphs is  $(1-1/e + O(1/\sqrt{d}))$.  
\label{thm:intgap:bdeg}
\end{theorem} 

\begin{proof}
To prove this theorem, we construct a ``bounded degree'' analogue of our good gadgets. 
That is, the $(M,k)$-good gadget in this setting guarantees that 

\begin{itemize} 
\item All leaves of $T_i$ are in the same layer $L_h$. 

\item For each layer $L_j$, we have $|\sset \cap L_j | \leq k$. For each tree $T_i$, for each leaf $v \in V(T_i)$, the unique path from $r_i$ to $v$ contains one vertex in $\sset$. 

\item For any subset $\bset \subseteq \set{r_1,\ldots,r_M}$, for any strategy $\uset$, at least $(1-1/k - O(1/d))\frac{|\bset||L_h|}{M}$ vertices in $L_h$ are $(\uset, \bset)$-risky. 
\end{itemize}  

This gadget can be used to recursively construct the instance in $k$ phases.
The final instance guarantees the integrality gap of $1- (1-1/k)^k + O(k/d)$. 
By setting $k = \sqrt{d}$, we get the integrality gap of $(1-1/e + O(1/\sqrt{d}))$ as desired\footnote{By analyzing the Taylor's series expansion of $1/e - (1-1/k)^k$, we get the term $\frac{1}{2e k} + O(1/k^2)$}.   
\end{proof}
\end{appendixextra}

\section{Hartke's Constraints}

\label{sec:newlp}

\newcommand{\vset}{{\mathcal{V}}}
Due to the integrality gap result in the previous section, there is no hope to improve the best known algorithms via the canonical LP relaxation. 
Hartke~\cite{Hartke2006} suggested adding the following constraints to narrow down the integrality gap of the LP. 
\[\sum_{ u\in P_v \cup (T_v \cap L_j)} x_u \leq 1 \mbox{ for all vertex $v \in V(T)$ and layer $L_j$ below the layer of $v$} \] 

We write the new LP with these constraints below: 
\vspace{0.1in}

\noindent
\framebox[\textwidth]{ 
\begin{minipage}[b]{0.99\textwidth}
\vspace{-1em}
\begin{eqnarray*} 
 \mbox{(LP')} \\
  &\max & \sum_{v \in V} y_v\\
  && \sum_{u \in P_v \cup (T_v \cap L_j)} x_u \leq 1 \mbox{ for all layer $j$ below vertex $v$} \\ 
  &&  y_v \leq \sum_{u \in P_v} x_u \mbox{ for all $v\in V$}\\
  && x_v, y_v \in [0,1] \mbox{ for all $v$ }  
\end{eqnarray*}
\end{minipage}
}

\begin{proposition} 
\label{prop:setyvalues}
Given the values $\set{x_v}_{v \in V(T)}$ that satisfy the first set of constraints, then the solution $(x,y)$ defined by $y_v = \sum_{u \in P_v} x_v$ is feasible for (LP') and at least as good as any other feasible $(x,y')$.  
\end{proposition} 

In this section, we study the power of this LP and provide three evidences that it may be stronger than (LP-1). 

\subsection{New properties of extreme points}
In this section, we show that Finbow et al.~tractable instances~\cite{FinbowG09} admit a polynomial time exact algorithm via (LP') (in fact, any optimal extreme point for (LP') is integral.) 
In contrast, we show that (LP-1) contains an extreme point that is not integral.

We first present the following structural lemma. 
\begin{lemma}
\label{lem: extreme}  
Let $({\bf x},{\bf y})$ be an optimal extreme point for (LP') on instance $T$ rooted at $s$. Suppose $s$ has two children, denoted by $a$ and $b$.  
Then $x_a, x_b \in \set{0,1}$.  
\end{lemma}  

\begin{appendixproof}{sec:newlp}{\subsection{Proof of Lemma \ref{lem: extreme}}}
Suppose that $x_a, x_b \in (0,1)$. 
We will define two solutions $({\bf x'},{\bf y'})$ and $({\bf x''}, {\bf y''})$ and derive that $({\bf x},{\bf y})$ can be written as a convex combination of $({\bf x'},{\bf y'})$ and $({\bf x''}, {\bf y''})$, a contradiction.

First, we define $({\bf x}',{\bf y}')$ by setting $x'_b = 1, x'_a=0$. For each vertex $v \in T_b$, we set $x'_v = 0$. For each vertex $v \in T_a$, we define $x'_v = x_v/(1-x_a)$.
We verify that $x'$ is feasible for (LP'): For each $v \in T_a$ and any layer $L_j$ below $v$, $\sum_{u \in P_v} x'_u + \sum_{u \in T_v \cap L_j} x'_u = \frac{\paren{\sum_{u \in P_v} x_u} - x_a}{(1-x_a)} + \frac{\sum_{u \in T_v \cap L_j} x_u}{(1-x_a)} \leq \frac{\paren{\sum_{u \in P_v \cup (T_v \cap L_j)} x_u} -x_a}{(1-x_a)} \leq 1$ (the last inequality is due to the fact that ${\bf x}$ is feasible).
The constraint is obviously satisfied for all $v \in T_b$.  
For the root node $v =s$, we have $\sum_{u \in L_j} x'_u = \frac{\paren{\sum_{u\ \in (L_j \cap T_a)} x_u} - x_a}{(1-x_a)} \leq 1$.  

We define $({\bf x}'',{\bf y}'')$ analogously: $x''_b =0, x''_a =1$. For each vertex $v \in T_a$, we set $x''_v = 0$, and for each $v \in T_b$, we define $x''_v = x_v/ (1-x_b)$. 
It can be checked similarly that $({\bf x}'',{\bf y}'')$ is a feasible solution.  

\begin{claim} 
If ${\bf x}$ is an optimal extreme point, then $x_a + x_b =1$. 
\label{claim:extremepoint}
\end{claim} 

\begin{proof}Observe that, for each $v \in T_b$, $y'_v = 1$ and for each $v \in T_a$, $y'_v = \frac{y_v  - x_a}{1-x_a}$.  
The objective value of ${\bf x'}$ is $|T_b| + \sum_{v \in T_a} y'_v = |T_b| + \frac{1}{(1-x_a)} \sum_{v \in T_a} (y_v - x_a) = |T_b| + \frac{\sum_{v \in T_a} y_v}{(1-x_a)} - \frac{x_a}{(1-x_a)} |T_a|$.
Similarly, the objective value of solution ${\bf x''}$ is $|T_a| + \frac{1}{(1-x_b)} \sum_{v \in T_b} (y_v - x_b) = |T_a| + \frac{\sum_{v \in T_b} y_v}{(1-x_b)} - \frac{x_b}{(1-x_b)} |T_b|$.
 
Consider the convex combination $\frac{1-x_a}{(2-x_a-x_b)} {\bf x'} + \frac{1-x_b}{(2-x_a-x_b)} {\bf x''}$. 
This solution is feasible and has the objective value of 
$$\frac{1}{(2-x_a-x_b)}\cdot \paren{ (1-x_a -x_b) \paren{|T_a| + |T_b|}+ \sum_{v \in V(T)} y_v}$$
If $x_a + x_b <1$, we apply the fact that $|T_a| + |T_b| > \sum_{v \in V(T)} y_v$ to get the objective of strictly more than $\sum_{v \in V(T)} y_v$, contradicting the fact that $({\bf x}, {\bf y})$ is optimal.  
\end{proof} 

Finally, we define the convex combination by ${\bf z}= (1-x_a) {\bf x'}  + x_a {\bf x''}$.
It can be verified easily that $z_v = x_v$ for all $v \in V(T)$.  
\end{appendixproof} 

\vspace{0.1in} 

\noindent {\bf Finbow et al. Instances:} In this instance, the tree has degree at most $3$ and the root has degree $2$.  
Finbow et al.~\cite{FinbowG09} showed that this is polynomial time solvable. 

\begin{wrapfigure}[11]{r}{0.25\textwidth}
\centering
\includegraphics{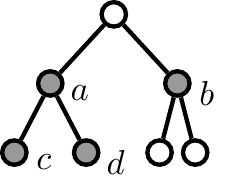}
\caption{Instance with a non-integral extreme point for (LP-1). Gray vertices: $x_v=1/2$; otherwise: $x_v = 0$.}
\label{fig:lp1badcase}
\end{wrapfigure}

\begin{theorem} 
Let $(T,s)$ be an input instance where $T$ has degree at most $3$ and $s$ has degree two. Let $(x,y)$ be a feasible fractional solution for (LP-3). 
Then there is a polynomial time algorithm that saves at least $\sum_{v \in V(T)} y_v$ vertices. 
\label{thm:finbowinst}
\end{theorem}  
\begin{appendixproof}{sec:newlp}{\subsection{Proof of Theorem \ref{thm:finbowinst}}}
We prove this by induction on the number of nodes in the tree that, for any tree $(T',s')$ that is a Finbow et al. instance, for any fractional solution $(x,y)$ for (LP'), there is an integral solution $(x',y')$ such that $\sum_{v \in T' \setminus \set{s'}} y'_v = \sum_{v \in T' \setminus \set{s'}} y_v$.  
Let $a$ and $b$ be the children of the root $s$.
From Lemma~\ref{lem: extreme}, assume w.l.o.g. that $x_a = 1$, so we have $\sum_{v \in T_a} y_v = |T_a|$. 
By the induction hypothesis, there is an integral solution $(x',y')$ for the subtree $T_b$ such that $\sum_{v \in T_b} y'_v = \sum_{v \in T_b \setminus \set{b}} y'_v = \sum_{v \in T_b} y_v$.  
The solution $(x',y')$ can be extended to the instance $T$ by defining $x'_a = 1$. 
This solution has the objective value of $|T_a| + \sum_{v \in T_b} y'_b = |T_a| + \sum_{v \in T_b} y_b$, completing the proof. 
\end{appendixproof}

\noindent {\bf Bad instance for (LP-1):} 
We show in Figure \ref{fig:lp1badcase} a Finbow et al. instance as well as a
solution for (LP-1) that is optimal and an extreme point, but not integral.

\begin{appendixextra}{sec:newlp}{\subsection{Analysis of Figure \ref{fig:lp1badcase}}}
\begin{claim}

The solution $(x,y)$ represented in Figure \ref{fig:lp1badcase}, with $y$ defined according to
Proposition \ref{prop:setyvalues}, is an extreme point of this instance for
(LP-1).

\end{claim}

\begin{proof}

Suppose (for contradiction) that $(x,y)$ is not an extreme point. Then, there
are distinct solutions $(x',y')$, $(x'',y'')$ and $\alpha \in (0,1)$ such that
$(x,y) = \alpha(x',y') + (1-\alpha)(x'',y'')$. Since $y_c = 1$ and $y'_c,
y''_c \leq 1$, then $y'_c = y''_c = 1$, and likewise, $y'_d = y''_d = 1$.
Combining that $x'_a + x'_c = y'_c = 1$ with $x'_a + x'_d = y'_d = 1$ and
$x'_c + x'_d \leq 1$, we conclude that $x'_a \geq 1/2$. Similarly, we get that
$x''_a \geq 1/2$, which implies that $x'_a = x''_a = 1/2$.

Similar reasoning using that $x'_a + x'_b \leq 1$ allows us to conclude that
$x'_b = x''_b = 1/2$, and thus, $(x',y') = (x'',y'') = (x,y)$, which
contradicts our assumption.
\end{proof}
\end{appendixextra}

\subsection{Rounding \texorpdfstring{$1/2$}{1/2}-integral Solutions}
\label{sec:half-int}
We say that the LP solution $(x,y)$ is $(1/k)$-integral if, for all $v$, we have $x_v = r_v/k$ for some integer $r_v
\in \{0,\ldots,k\}$.  
By standard LP theory, one can assume that the LP solution is $(1/k)$-integral for some polynomially large integer $k$.

In this section, we consider the case when $k=2$ ($1/2$-integral LP solutions).  
From Theorem~\ref{thm: gap 1/k}, (LP-1) is not strong enough to obtain a $3/4+\epsilon$ approximation algorithm, for any $\epsilon > 0$. 
Here, we show a $5/6$ approximation algorithm based on rounding (LP').  

\begin{theorem} 
Given a solution $(x,y)$ for (LP') that is $1/2$-integral, there is a polynomial time algorithm that produces a solution of cost $5/6\,\sum_{v \in V(T)} y_v$.  
\label{thm:12intalgo}
\end{theorem} 

We believe that the extreme points in some interesting special cases will be $1/2$-integral. 

\vspace{0.1in} 

\noindent {\bf Algorithm's Description:} 
Initially, $\uset = \emptyset$. 
Our algorithm considers the layers $L_1,\ldots, L_n$ in this order.
When the algorithm looks at layer $L_j$, it picks a vertex $u_j$ and adds it to $\uset$, as follows. 
Consider $A_j \subseteq L_j$, where $A_j = \set{v \in L_j: x_v >0}$.
Let $A'_j \subseteq A_j$ contain vertices $v$ such that there is no ancestor of $v$ that belongs to $A_{j'}$ for some $j' <j$, and $A''_j = A_j \setminus A'_j$, i.e. for each $v \in A''_j$, there is another vertex $u \in A_{j'}$ for some $j' <j$ such that $u$ is an ancestor of $v$.
We choose the vertex $u_j$ based on the following rules: 
\begin{itemize} 
\item If there is only one $v \in A_j$, such that $v$ is not saved by $\uset$ so far, choose $u_j=v$. 

\item Otherwise, if $|A'_j| =2$, pick $u_j$ at random from $A'_j$ with uniform probability.
Similarly, if $|A''_j| =2$, pick $u_j$ at random from $A''_j$.  

\item Otherwise, we have the case $|A'_j| = |A''_j| = 1$. In this case, we pick vertex $u_j$ from $A'_j$ with probability $1/3$; otherwise, we take from $A''_j$.  
\end{itemize}

\begin{appendixextra}{sec:newlp}{\subsection{Analysis of Theorem \ref{thm:12intalgo}}}
\textbf{Analysis: }    
Below, we argue that each vertex $v \in V(T): x_v >0$ is saved with probability at least $(5/6) y_v$. 
It is clear that this implies the theorem: Consider a vertex $v': x_{v'}=0$.
If $y_{v'}=0$, we are immediately done. Otherwise, consider the bottommost ancestor $v$ of $v'$ such that $x_v >0$. 
Since $y_v = y_{v'}$, the probability that $v'$ is saved is the same as that of $v$, which is at least $(5/6) y_v$.    

We analyze a number of cases.
Consider a layer $L_j$ such that $|A_j|=1$. 
Such a vertex $v \in A_j$ is saved with probability $1$. 

Next, consider a layer $L_j$ such that $|A'_j| =2$. 
Each vertex $v \in A'_j$ is saved with probability $1/2$ and $y_v = 1/2$. 
So, in this case, the probability of saving $v$ is more than $(5/6) y_v$.  

\begin{lemma}
\label{one-and-one} 
Let $L_j$ be the layer such that $|A'_j| = |A''_j| =1$. 
Then the vertex $u \in A'_j$ is saved with probability $2/3 \geq (5/6) y_u$ and vertex $v \in A''_j$ is saved with probability $5/6$. 
\end{lemma} 

\begin{proof} 
Let $v' \in A_{j'}$ be the ancestor of $v$ in some layer above $A_j$. 
The fact that $v$ has not been saved means that $v'$ is not picked by the algorithm, when it processed $A_{j'}$. 

We prove the lemma by induction on the value of $j$. 
For the base case, let $L_j$ be the first layer such that $|A'_j| = |A''_j| =1$.
This means that the layer $L_{j'}$ must have $|A'_{j'}| = 2$, and therefore the probability of $v'$ being saved is at least $1/2$. 
Vertex $u$ is not saved only if both $v'$ and $u$ are not picked, and this happens with probability $1/2 \cdot 2/3 = 1/3$.
Hence, vertex $u$ is saved with probability $2/3$ as desired.   
Consider now the base case for vertex $v$, which is not saved only if $v'$ is not saved and $u$ is picked by the algorithm among $\set{u,v}$.
This happens with probability $1/2 \cdot 1/3 = 1/6$, thus completing the proof of the base case.

 For the purpose of induction, we now assume that, for all layer $L_i$ above $L_j$ such that $|A'_{i}| = |A''_i| =1$, the probability that the algorithm saves the vertex in $A'_i$ is at least $2/3$.
Since the vertex $u$ is not saved only if $v'$ is not saved, this probability is either $1/2$ or $1/3$ depending on the layer to which $v'$ belongs. If it is $1/3$, we are done; otherwise, the probability is at most $1/2 \cdot 2/3 = 1/3$.   
Now consider vertex $v$, which is not saved only if $v'$ is not saved and $u$ is picked at $L_j$. This happens with probability at most $1/2 \cdot 1/3 = 1/6$.  
\end{proof}

\begin{lemma}
Let $L_j$ be a layer such that $A''_j = \set{u,v}$ (containing two vertices). Then each such vertex is saved with probability at least $5/6$. 
\end{lemma}

\begin{proof}

Let $u'$ and $v'$ be the ancestors of $u$ and $v$ in some sets $A'_{i}$ and $A'_k$ above the layer $L_j$.
There are the two possibilities: either both $u'$ and $v'$ are in layers with $|A'_i| = |A'_k| =2$ (maybe $i=k$); or $u'$ is in the layer with $|A'_i| = |A''_i| =1$. We remark that $u' \neq v'$: otherwise, the LP constraint for $v'$ and $L_j$ would not be satisfied.

For $u$ or $v$ to be unsaved, we need that both $u'$ and $v'$ are not saved by the algorithm.  
Otherwise, if, say, $u'$ is saved, $u$ is also saved, and the algorithm would have picked $v$.

\begin{align*}
P[u \text{ is not saved}] &= P[u \text{ not picked} \wedge u' \text{ is not saved} \wedge v' \text{ is not saved}] \\
                   &= P[u \text{ not picked}]\cdot P[u' \text{ is not saved} \wedge v' \text{ is not saved}] \\
                   &= \frac12 \cdot \frac14 
                    = \frac18\\
P[u \text{ is saved}] &=\frac78 \geq \frac56
\end{align*}

It must be that $P[u' \text{ burns} \wedge v' \text{ burns}] \leq 1/4$,
since either $u'$ and $v'$ are in different layers or they are in the same
layer. 
If they are in different layers, picking each of them is independent,
and the probability of neither being saved is at most $1/4$. 
If they are in
the same layer, one of them is necessarily picked, which implies that the
probability of neither being saved is $0$. In any case, the probability is at most $1/4$.

In the second case, at least one of the vertices $u'$, $v'$ is in a layer with one
2-special vertex. W.~l.~o.~g.~let $u'$ be in such a
layer. By Lemma~\ref{one-and-one}, we know that the probability that $u'$
is not saved is at most $1/3$. 
Therefore, 

\begin{align*}
P[u \text{ burns}] &= P[u \text{ not picked} \wedge u' \text{ burns} \wedge v' \text{ burns}] \\
                   &= P[u \text{ not picked}]\cdot P[u' \text{ burns} \wedge v' \text{ burns}] \\
                   &\leq P[u \text{ not picked}]\cdot P[u' \text{ burns}] \\
                   &\leq \frac12 \cdot \frac13 
                    = \frac16\\
P[u \text{ is saved}] &\geq \frac56
\end{align*}

The proof for both cases works analogously for $v$.
\end{proof}

\end{appendixextra}

\subsection{Ruling Out the Gap Instances in Section~\ref{sec:gap}}  
\label{sec:well-separable}

In this section, we show that the integrality gap instances for (LP-1) presented in the previous section admit a better than $(1-1/e)$ approximation via (LP').
To this end, we introduce the concept of well-separable LP solutions  and show an improved rounding algorithm for solutions in this class.

Let $\eta \in (0,1)$. 
Given an LP solution $(x,y)$ for (LP-1) or (LP'), we say that a vertex $v$ is $\eta$-light if $\sum_{u \in P_v \setminus \set{v}} x_u < \eta$; if a vertex $v$ is not $\eta$-light, we say that it is $\eta$-heavy. 
A fractional solution is said to be $\eta$-separable if for all layer $j$, either all vertices in $L_j$ are $\eta$-light, or they are all $\eta$-heavy.  
For an $\eta$-separable LP solution $(x,y)$, each layer $L_j$ is either an $\eta$-light layer that contains only $\eta$-light vertices, or $\eta$-heavy layer that contains only $\eta$-heavy vertices.

\begin{observation} 
The LP solution in Section~\ref{sec:gap} is $\eta$-separable for all values of $\eta \in \set{1/k, 2/k, \ldots, 1}$.  
\end{observation}   

\begin{theorem}
If the LP solution $(x,y)$ is $\eta$-separable for some $\eta$, then there is an efficient algorithm that produces an integral solution of cost $(1-1/e + f(\eta))\sum_{v} y_v$, where $f(\eta)$ is some function depending only on $\eta$.
\label{thm:wellseparable}
\end{theorem}  

\vspace{0.1in} 

\textbf{Algorithm:}
Let $T$ be an input tree, and $(x,y)$ be a solution for (LP') on $T$ that is $\eta$-separable for some constant $\eta \in (0,1)$.
Our algorithm proceeds in two phases.  
In the first phase, it performs
randomized rounding independently for each $\eta$-light layer.
Denote by $V_1$ the (random) collection of vertices selected in this phase. 
Then,
in the second phase, our algorithm performs randomized rounding conditioned on the solutions in
the first phase. 
In particular, when we process each $\eta$-heavy layer $L_j$, let $\tilde L_j$ be the collection of vertices that have not yet been saved by $V_1$.
We sample one vertex $v \in \tilde L_j$ from the distribution $\left\{\frac{x_v}{x(\tilde L_j)}\right\}_{v \in \tilde L_j}$. 
Let $V_2$ be the set of vertices chosen from the second phase. 
This completes the description of our algorithm. 
\begin{appendixextra}{sec:newlp}{\subsection{Analysis of Theorem \ref{thm:wellseparable}}}

For notational simplification, we present the proof when $\eta = 1/2$. It will be relatively obvious that the proof can be generalized to work for any $\eta$.  
Now we argue that each terminal $t \in \xset$ is saved with probability at least $(1-1/e +\delta)y_t$ for some universal constant $\delta >0$ that depends only on $\eta$. 
We will need the following simple observation that follows directly by standard probabilistic analysis.

\begin{proposition}
\label{prop:LP-frac}  
For each vertex $v \in V(T)$, the probability that $v$ is not saved is at most $\prod_{u \in P_v}(1-x_u) \geq 1- e^{-y_v}$. 
\end{proposition} 
 
We start by analyzing two easy cases.  

\begin{lemma} 
\label{newlp:specialcases}
Consider $t \in \xset$. If $y_t < 0.9$ or there is some ancestor $v \in P_t$ such that $x_v > 0.2$, then the probability that $v$ is saved by the algorithm is at least $(1-1/e+\delta) y_t$.  
\end{lemma}  
\begin{proof} 
First, let us consider the case where $y_t < 0.9$. The probability of $t$
being saved is at least $1-e^{-y_v}$, according to the straightforward analysis.
If $y_t < 0.9$, we have $1-e^{-y_t}/y_t > 1.04(1-1/e)y_t$ as desired. 

Consider now the second case when $x_v > 0.2$ for some ancestor $v \in P_t$. 
The bound used
typically in the analysis is only tight when the values are all small, and,
therefore, we get an advantage when one of the values is relatively big. In
particular,
\begin{align*}
\pr{}{t \text{ is saved}} &\geq 1 - \prod_{u \in P_t} (1-x_u) \\
&\geq 1 - (1-x_v)e^{-(y_t-x_v)} \\
&\geq 1 - (1-0.2)e^{-(y_t-0.2)} \\
&\geq 1.01(1-1/e)y_v
\end{align*}
\end{proof} 

From now on, we only consider those terminals $t \in\xset$ such that $y_t \geq 0.9$ and $x_v < 0.2$, for all $v \in
P_t$. 
We remark here that if the value of $\eta$ is not $1/2$, we can easily pick other suitable thresholds instead of $0.9$ and $0.2$.  

For each vertex $v \in V$, let $\xset_1 \subseteq \xset$ be the set of terminals that are saved by $V_1$, i.e. a vertex $t \in \xset_1$ if and only if $t$ is a descendant of some vertex in $V_1$.  
Let $\xset_2 \subseteq \xset \setminus \xset_1$ contain the set of terminals that are not saved by the first phase, but are saved by the second phase, i.e. $t \in \xset_2$ if and only if $t$ has some ancestor in $V_2$.  
\[\pr{V_1, V_2}{t \not \in \xset_1 \cup \xset_2} = \pr{V_1, V_2}{t \not \in \xset_1}
                                  \pr{V_1, V_2}{t \not \in \xset_2 : t \not \in \xset_1}\]

\newcommand{\avgsumgoodterm}[0]{\ensuremath{\sum_{i \in \lset_{t,2}}x_{v_{i,t}} x(\tilde L_{i,t})}}

For any terminal $t$, let $\sset'_t$ and $\sset''_t$ be the sets of ancestors of $t$ that are $\eta$-light and $\eta$-heavy respectively, i.e. ancestors in $\sset'_t$ and $\sset''_t$ are considered by the algorithm in Phase 1 and 2 respectively. 
By Proposition~\ref{prop:LP-frac}, we can upper bound the first term by $e^{-x(\sset'_t)}$. 
In the rest of this section, we show that the second term is upper bounded by $e^{-x(\sset''_t)} c$ for some $c <1$, and therefore $\pr{}{t \not \in \xset_1 \cup \xset_2} \leq c e^{-x(\sset'_t) - x(\sset''_t)} \leq c e^{-y_t}$, as desired.   

The following lemma is the main technical tool we need in the analysis.
We remark that this lemma is the main difference between (LP') and (LP-2).  
\begin{lemma}
\label{newlp:advantage}
Let $t \in \xset$ and  $L_j$ be a
layer containing some $\eta$-heavy ancestor of $t$.  
Then 
\[{\mathbb E}_{V_1} [x(\tilde L_j) \mid t \not \in \xset_1] \leq  \alpha \]
for $\displaystyle{\alpha = \frac12 + (1-e^{-1/2}) \leq 0.9}$ 
\end{lemma}

Intuitively, this lemma says that any terminal that is still not saved by the result of the first phase will have a relatively ``sparse'' layer above it.   
We defer the proof of this lemma to the next subsection. 
Now we proceed to complete the analysis.

For each vertex $v$, denote by $\ell(v)$ the layer to which vertex $v$ belongs. 
For a fixed choice of $V_1$, we say that terminal $t$ is {\em partially protected by $V_1$} if 
$\sum_{ v \in \sset''_t} x_v x(\tilde L_{\ell(v)}) \leq C x(\sset''_t)$ (we will choose the value of $C \in (\alpha,1)$ later).
Let $\xset' \subseteq \xset \setminus \xset_1$ denote the subset of terminals that are partially protected by $V_1$.  

\begin{claim}
\label{improv:layer:alt:goodterm}
For any $t \in \xset$,
$\pr{V_1}{t \in \xset' \mid t \not \in \xset_1} \geq 1- \alpha/C$. 
\end{claim}

\begin{proof} 
By linearity of expectation and Lemma \ref{newlp:advantage}, 
\[\expect{V_1}{\sum_{v \in \sset''_t} x_v x(\tilde L_{\ell(v)}) \mid t \not \in \xset_1} = \sum_{v \in \sset''_t} x_v \expect{V_1}{x(\tilde L_{\ell(v)}) \mid t \not \in \xset_1} \leq \alpha x(\sset''_t)\]
Using Markov's inequality, 
\begin{align*}
&\pr{V_1}{\sum_{v \in \sset''_t} x_v x(\tilde L_{\ell(v)}) \leq C x(\sset''_t) \mid t \not \in \xset_1} \\
  &\quad\quad\quad\quad= 1 - \pr{}{\sum_{v \in \sset''_t} x_v x(\tilde L_{\ell(v)})   > C x(\sset''_t) \mid t \not \in \xset_1} \\
  &\quad\quad\quad\quad\geq 1 - \frac{\alpha x(\sset''_t)}{C x(\sset''_t)} \\
  &\quad\quad\quad\quad= 1 - \frac{\alpha }{C}
\end{align*}
\end{proof}

We can now rewrite the probability of a terminal $t \in \xset$ not being saved by the solution after the second phase.
\begin{align*}
&\pr{V_1, V_2}{t \not \in \xset_2 \mid  t \not \in \xset_1 } \\
&= \pr{}{ t \in \xset' \mid  t \not \in \xset_1}\,\pr{}{t \not \in \xset_2\mid  t \in \xset'} 
        + \pr{}{t \not \in \xset' \mid t \not \in \xset_1}\,\pr{}{ t \not \in \xset_2 \mid t \not \in \xset'}\\
&\leq (1-\alpha/C) \pr{V_1, V_2}{t \not \in \xset_2  \mid t \in \xset'} 
       + \frac{\alpha}{C} \cdot e^{-x(\sset''_t)}
\end{align*}

The last inequality holds because $\pr{V_1, V_2}{t \not \in \xset_2 \mid t \not\in \xset'}$ is at most $e^{-x(\sset''_t)}$ from Proposition~\ref{prop:LP-frac}.
 
It remains to provide a better
upper bound for $\pr{}{t \not \in \xset_2 \mid t \in \xset'}$. 
Consider a vertex $v \in \sset''_t$ that is involved in the second phase rounding.  
We say that vertex $v$  is \emph{good} for $t$ and $V_1$ if $x(\tilde L_{\ell(v)}) \leq C'$ (we will choose the value $C' \in (C,1)$ later.)  
Denote by $\sset^{good}_t \subseteq \sset''_t$ the set of good ancestors of $t$. 
The following claim ensures that good ancestors have large LP-weight in total. 

\begin{claim}
For any node $t \in \xset'$, $x(\sset^{good}_t) =  \sum_{v \in \sset^{good}_t} x_{v} \geq (1-C/C') x(\sset''_t)$. 
\end{claim}

\begin{proof}
Suppose (for contradiction) that the fraction of good layers was less than $1-C/C'$. This means that $x(\sset''_t \setminus \sset^{good}_t) \geq C/C'$.
For each such $v \in \sset''_t \setminus \sset^{good}_t$, we have $x(\tilde L(v)) > C'$.  
Then,
$\sum_{v \in \sset''_t} x_v x(\tilde L_{\ell(v)}) > \sum_{v\in \sset''_t \setminus \sset^{good}_t} x_v C' \geq C$. 
This contradicts the assumption that $t$ is partially protected, and concludes our proof.
\end{proof}

Now the following lemma follows.  

\newcommand{\qqqq}{\quad\quad\quad\quad}
\begin{lemma}
\label{lem:partial}  
$\pr{V_1, V_2}{t \not \in \xset_2 \mid t \in \xset'} \leq  e^{-x(\sset''_t)} e^{-(1-\frac{C}{C'}) x(\sset''_t) \paren{\frac{1}{C'}-1}}$ 
\end{lemma}

\begin{proof} 
\begin{align*}
&\pr{V_1, V_2}{t \not \in \xset_2 \mid t \in \xset'} \\
&\qqqq=    \sum_{V^\prime_1: t \in \xset'} \pr{V_1}{V_1 = V^\prime_1}\,\pr{V_2}{t \not \in \xset_2 \mid V_1 = V^\prime_1} \\
&\qqqq\leq \sum_{V^\prime_1: t \in \xset'} \pr{V_1}{V_1 = V^\prime_1}\,
          \prod_{\text{bad }v \in \sset''_t}\paren{1-x_{v}}\,
          \prod_{\text{good }v \in \sset''_t}\paren{1-\frac{x_{v}}{C'}} \\
&\qqqq\leq \sum_{V^\prime_1: t \in \xset'} \pr{V_1}{V_1 = V^\prime_1}\,
          \prod_{\text{bad } v \in \sset''_t} e^{-x_{v}}\,
          \prod_{\text{good }v \in \sset''_t} e^{-x_{v}/C'} \\
&\qqqq\leq \sum_{V^\prime_1: t \in \xset'} \pr{V_1}{V_1 = V^\prime_1}\,
          e^{-x(\sset''_t)\frac{C}{C'} - (1-C/C') x(\sset''_t) / C'} \\
&\qqqq\leq e^{-x(\sset''_t)\frac{C}{C'} - (1-C/C')x(\sset''_t) / C'}\,
          \sum_{V^\prime_1: t \in \xset'} \pr{V_1}{V_1 = V^\prime_1} \\
&\qqqq\leq e^{-x(\sset''_t)\frac{C}{C'} - (1-C/C') x(\sset''_t) / C'} \\
&\qqqq\leq e^{-x(\sset''_t)} e^{-(1-C/C') x(\sset''_t) \paren{\frac1{C'}-1}}
\end{align*}
\end{proof} 
 
Now we choose the parameters $C$ and $C'$ such that $C = (1+\delta) \alpha$, $C' = (1+\delta)C$, and $(1+\delta)C' = 1$, where $(1+\delta)^3 = 1/\alpha$.
Notice that this choice of parameters satisfy our previous requirements that $\alpha < C < C' <1$. 
The above lemma then gives the upper bound of $e^{-x(\sset''_t)} e^{-\frac{\delta^2}{1+\delta} x(\sset''_t)}$, which is at most $e^{-(1+\delta^2/2) x(\sset''_t)}$.  
Since $\delta>0$ is a constant, notice that we do have an advantage over the standard LP rounding in this case.
Now we plug in all the parameters to obtain the final result.   

\begin{align*}
\pr{V_1, V_2}{t \not \in \xset_1 \cup \xset_2} 
&= \pr{V_1, V_2}{t \not\in \xset_1}\,\pr{V_1, V_2}{t \not \in \xset_2 \mid t \not\in \xset_1} \\
&\leq e^{-x(\sset'_t)}\,\paren{(1-\alpha/C) \, \pr{V_1, V_2}{t \not \in \xset_2 \mid t \in \xset'} + \frac{\alpha}{C} \,e^{-x(\sset''_t)}} \\
&\leq e^{-x(\sset'_t)}\,
        \paren{(1-\alpha/C) \, e^{-x(\sset''_t)} e^{-\frac{\delta^2}{2} x(\sset''_t)}  + \frac{\alpha}{C} \, e^{-x(\sset''_t)}} \\
&\leq e^{-y_t}\,\paren{(1-\alpha/C) \, e^{- \frac{\delta^2}{2}x(\sset''_t) } + \alpha/C } \\
& \leq  e^{-y_t} \paren{\frac{\delta}{1+\delta} e^{-(1+\delta^2/2)x(\sset''_t)} + \frac{1}{1+\delta}} 
\end{align*}

Since we assume that $y_t > 0.9$ and $x_v \leq 0.2$, we must have $x(\sset''_t) \geq 0.2$,
 and therefore the above term can be seen as $e^{-y_t} \cdot \delta'$ for some $\delta' <1$.
Overall, the approximation factor we get is $(1-\delta'/e)$ for some universal constant $\delta' \in (0,1)$.  

\subsubsection{Proof of Lemma~\ref{newlp:advantage}}
For each $u$, let $\eset_u$ denote the event that $u$ is not saved by $V_1$.  
First we break the expectation term into $\sum_{u \in L_j} x_u \pr{}{\eset_u \mid t \not \in \xset_1}$.  
Let $v \in L$ be the ancestor of $t$ in layer $L_j$.  
We break down the sum further based on the ``LP coverage'' of the least common ancestor between $u$ and $v$, as follows: 
\[\sum_{i=0}^{k/2} \sum_{u\in L_j: q'(lca(u,v))=i} x_u \pr{}{\eset_u \mid t \not\in \xset_1} \]

Here, $q'(u)$ denotes $k \cdot x(P_u)$; this term is integral since we consider the $1/k$-integral solution $(x,y)$.  
The rest of this section is devoted to upper bounding the term $\pr{}{\eset_u \mid t \not \in \xset_1}$. 
The following claim gives the bound based on the level $i$ to which the least common ancestor belongs.

\begin{claim} 
\label{newlp:advantage:probpair}
For each $u \in L_j$ such that $q'(lca(u,v)) = i$, 
\[\pr{}{\eset_u \mid t \not \in \xset_1} \leq e^{-\paren{1/2 -i/k}}\]
\end{claim}

\begin{proof}
First, we recall that $y_u \geq 1/2$ and $q'(u) \geq k/2$, since $u$ is in the $1/2$-heavy layer $L_j$.  
Let $w= lca(u,v)$ and $P'$ be the path that connects $w$ to $u$.
Moreover, denote by $S \subseteq P'$ the set of light vertices on the path $P'$, i.e. $S = \sset'_t \cap P'$.   
Notice that $x(S) \geq \sum_{a \in \sset'_t \cap P_u} x_a - \sum_{a \in P_w} x_a  \geq (1/2 - i/k)$.  

For each $w' \in S$, $\pr{}{w' \not \in V_1 \mid t \not \in \xset_1} \in \set{1-x_{w^\prime}, 1- x_{w^\prime}/(1-x_{v^\prime})}$ depending on whether there is a vertex $v^\prime$ in $P_v$ that shares a layer with $w'$. 
In any case, it holds that $\pr{}{w' \not \in V_1 \mid t \not \in \xset_1} \leq (1-x_{w^\prime})$.
This implies that 
\begin{align*}
\pr{}{\eset_u \mid t \not \in \xset_1 } &\leq \prod_{w' \in S} \pr{}{w' \not \in V_1 \mid t \not \in \xset_1}  \\
&\leq \prod_{w' \in S} (1-x_{w^\prime}) \\
&\leq \prod_{w' \in S} e^{-x_{w^\prime}} \\
&\leq e^{-(1/2-i/k)} 
\end{align*}
\end{proof} 

\begin{claim} 
\label{newlp:advantage:nelems}
Let $i$ be an integer and $L^\prime \subseteq L_j$ be the set of vertices $u$ such that $q'(lca(u,v))$ is at least $i$. Then $x(L^\prime) \leq (k-i)/k$.
\end{claim} 

\begin{proof} 
This claim is a consequence of Hartke's constraints.
Let $v'$ be the topmost ancestor of $v$ such that $q'(v') \geq i$. 
We remark that all vertices in $L'$ must be descendants of $v'$, so it must be that $\sum_{w \in P_{v'}} x_w + x(L^\prime) \leq 1$. 
The first term is $i/k$, implying that $x(L') \leq (k-i)/k$.
\end{proof} 

Let $L^i_j \subseteq L_j$ denote the set of vertices $u$ whose least common ancestor $lca(u,v)$ satisfies $q'(lca(u,v)) = i$.
As a consequence of Claim \ref{newlp:advantage:nelems}, $\sum_{i' \geq i} x(L^{i'}_j) \leq (k-i)/k$.  
Combining this inequality with Claim \ref{newlp:advantage:probpair}, we get that
\[\expect{}{x(\tilde L_j) \mid t \not \in \xset_1} \leq \sum_{i=0}^{k/2} x(L_j^i) e^{-1/2+i/k}\]

This term is maximized when $x(L_j^{k/2}) =1/2$ and $x(L_j^i) = 1/k$ for all other $i=0,1,\ldots, k/2-1$. 
This implies that 
\[\expect{}{x(\tilde L_j) \mid t \not \in \xset_1} \leq 1/2 + \sum_{i=0}^{k/2-1} e^{-1/2+i/k}/k\]

Finally, using some algebraic manipulation and the fact that $1+x \leq e^x$, we get  
\begin{align*}
\expect{}{x(\tilde L_j) \mid t \not \in \xset_1} 
&\leq 1/2 + \sum_{i=0}^{k/2-1} e^{-1/2+i/k}/k \\
&= 1/2 + \frac{1}{k}\, e^{-1/k}\,\frac{1-e^{-1/2}}{1-e^{-1/k}} \\
&= 1/2 + (1-e^{-1/2})\, \frac1{e^{1/k}}\, \frac{1/k}{1-e^{-1/k}} \\
&= 1/2 + (1-e^{-1/2})\, \frac{1/k}{e^{1/k}-1} \\
&\leq 1/2 + (1-e^{-1/2})
\end{align*}

\end{appendixextra}

\subsection{Integrality Gap for (LP')}
\label{sec:gap-newlp}  
In this section, we present an instance where (LP') has an integrality gap of
$5/6+\eps$, for any $\eps > 0$.
Interestingly, this instance admits an optimal $\frac{1}{2}$-integral LP solution.  

{\bf Gadget:}
The motivation of our construction is a simple gadget represented in
Figure \ref{12intgap:gadget}. 
In this instance, vertices are either {\em special} (colored gray) or {\em regular}. 
This gadget has three properties of our interest: 

\begin{itemize}

\item If we assign an LP-value of $x_v= 1/2$ to every special vertex, then this is a feasible LP solution that ensures $y_u = 1$ for all leaf $u$.  

\item For any integral solution $\uset$ that does not pick any vertex in the first layer of this gadget, at most $2$ out of $3$ leaves of the gadget are saved. 

\item Any pair of special vertices in the same layer do not have a common ancestor inside this gadget. 

\end{itemize}

\begin{wrapfigure}[13]{r}{0.25\textwidth}
\centering
\includegraphics{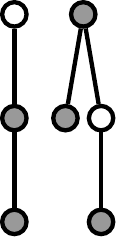}
\caption{Gadget used to get $5/6$ integrality gap. Special vertices are colored gray.}
\label{12intgap:gadget}
\end{wrapfigure}

Our integrality gap instance is constructed by creating partially overlapping copies of this gadget.
We describe it formally below.

\vspace{0.1in} 

\noindent {\bf Construction:}
The first layer of this instance, $L_1$, contains 4 nodes: two special nodes, which we
name $a(1)$ and $a(2)$, and two regular nodes, which we name $b(1)$ and $b(2)$.
We recall the definition of spider from Section~\ref{sec:good-gadget}.

Let $\alpha = 5\left\lceil 1/\eps \right\rceil$. The nodes $b(1)$ and $b(2)$ are
the roots of two spiders. 
Specifically, the spider $Z_1$ rooted at $b(1)$ has $\alpha$ feet, with one foot per layer, in consecutive layers $L_2,\ldots, L_{\alpha+1}$. For each $j \in [\alpha]$, denote by $b'(1,j)$, the $j^{th}$ foot of spider $Z_1$.  
The spider $Z_2$, rooted at $b(2)$, has $\alpha^2$ feet, with one foot per layer, in layers $L_{\alpha+2}, \ldots, L_{\alpha^2+\alpha+1}$. 
For each $j \in [\alpha^2]$, denote by $b'(2,j)$, the $j^{th}$ foot of spider $Z_2$.  
All the feet of spiders $Z_1$ and $Z_2$ are special vertices. 

For each $j \in [\alpha]$, the node $b'(1,j)$ is also the root of spider $Z'_{1,j}$, with $\alpha^2$ feet, lying in the $\alpha^2$ consecutive layers $L_{2+\alpha + j \alpha^2},\ldots, L_{1+\alpha+ (j+1) \alpha^2}$ (one foot per layer).
For $j' \in [\alpha^2]$, let $b''(1,j,j')$ denote the $j'$-th foot of spider $Z'_{1,j}$ that lies in layer $L_{1+\alpha+j\alpha^2 +j'}$.  
Notice that we have $\alpha^3$ such feet of these spiders $\set{Z'_{1,j}}_{j=1}^{\alpha}$ lying in layers $L_{2+\alpha+\alpha^2}, \ldots, L_{1+\alpha+\alpha^2+\alpha^3}$. 
Similarly, for each $j \in [\alpha^2]$, the node $b'(2,j)$ is the root of spider $Z'_{2,j}$ with $\alpha^2$ feet, lying in consecutive layers $L_{2+\alpha+\alpha^3 + j\alpha^2}, \ldots, L_{1+\alpha+ \alpha^3 + (j+1)\alpha^2}$. 
We denote by $b''(2,j,j')$ the $j'$-th foot of this spider.  

The special node $a(1)$ is also the root of spider $W_1$ which has $\alpha + \alpha^3$ feet: The first $\alpha$ feet, denoted by $a'(1,j)$ for $j \in [\alpha]$, are aligned with the nodes $b'(1,j)$, i.e. for each $j \in [\alpha]$, the foot $a'(1,j)$ of spider $W_1$ is in the same layer as the foot $b'(1,j)$ of $Z_1$. 
For each $j \in [\alpha], j' \in [\alpha^2]$, we also have a foot $a''(1,j,j')$ which is placed in the same layer as $b''(1,j,j')$.
Similarly, the special node $a(2)$ is the root of spider $W_2$ having $\alpha^2 + \alpha^4$ feet. 
For $j \in [\alpha^2]$, spider $W_2$ has a foot $a'(2,j)$ placed in the same layer as $b'(2,j)$. 
For $j \in [\alpha^2], j' \in [\alpha^2]$, $W_2$ also has a foot $a''(2,j,j')$ in the layer of $b''(2,j,j')$.  
All the feet of both $W_1$ and $W_2$ are special vertices.

Finally, for $i \in \set{1,2}$, and $j \in [\alpha^i]$, each node $a'(i,j)$ has $\alpha^{5-i}$ children, which are leaves of the instance.  
For $j \in [\alpha], j' \in [\alpha^2]$, the nodes $b''(i,j,j')$, $a''(i,j,j')$ have $\alpha^{3-i}$ children each which are also leaves of the instance. 
The set of terminals $\xset$ is simply the set of leaves. 

\begin{proposition} 
We have $|\xset|  = 6 \alpha^5$.
Moreover, 
(i) the number of terminals in subtrees $T_{a(1)} \cup T_{b(1)}$ is $3 \alpha^5$, and (ii) the number of terminals in subtrees $T_{a(2)} \cup T_{b(2)}$ is $3 \alpha^5$.  
\label{prop:56intgapanalysis}
\end{proposition}  

\begin{appendixproof}{sec:newlp}{\subsection{Analysis of Proposition \ref{prop:56intgapanalysis}}}
Each node $a'(1,j)$ has $\alpha^4$ children, and there are $\alpha$ such nodes. 
Similarly, each node $a'(2,j)$ has $\alpha^3$ children. There are $\alpha^2$ such nodes. 
This accounts for $2 \alpha^5$ terminals. 

For $i \in \set{1,2}$, each node $a''(i,j,j')$ has $\alpha^{3-i}$ children. 
There are $\alpha^{i+2}$ such nodes. 
This accounts for another $2 \alpha^5$ terminals. 
Finally, there are $\alpha^{3-i}$ children of each $b''(i,j,j')$, and there are $\alpha^{2+i}$ such nodes.

\end{appendixproof}

\vspace{0.1in} 

\noindent{\bf Fractional Solution:} 
Our construction guarantees that any path from root to leaf contains $2$ special vertices: For a leaf child of $a'(i,j)$, its path towards the root must contain $a'(i,j)$ and $a(i)$. 
For a leaf child of $a''(i,j,j')$, its path towards the root contains $a''(i,j,j')$ and $a(i)$. 
For a leaf child of $b''(i,j,j')$, the path towards the root contains $b''(i,j,j')$ and $b'(i,j)$.

\begin{lemma}
\label{lem:feas}  
For each special vertex $v$, for each layer $L_j$ below $v$, the set $L_j \cap T_v$ contains at most one special vertex. 
\end{lemma} 

\begin{appendixproof}{sec:newlp}{\subsection{Proof of Lemma \ref{lem:feas}}}
Each layer contains two special vertices of the form $\set{a'(i,j), b'(i',j')}$ or $\set{a''(i,j), b''(i',j')}$.
In any case, the least common ancestor of such two special vertices in the same layer is always the root $s$ (since one vertex is in $T_{a(i)}$, while the other is in $T_{b(i)}$)
This implies that, for any non-root vertex $v$, the set $L_j \cap T_v$ can contain at most one special vertex. 
\end{appendixproof}

Notice that, there are at most two special vertices per layer.  
We define the LP solution $x$, with $x_v =1/2$ for all special vertices $v$ and $x_v = 0$ for all other vertices. 
It is easy to verify that this is a feasible solution.  

We now check the constraint at $v$ and layer $L_j$ below $v$: If the sum $\sum_{u \in P_v} x_u = 0$, then the constraint is immediately satisfied, because $\sum_{ u \in L_j \cap T_v} x_u \leq 1$. 
If $\sum_{u \in P_v} x_u = 1/2$, let $v'$ be the special vertex ancestor of $v$. 
Lemma~\ref{lem:feas} guarantees that $\sum_{u \in L_j \cap T_v} x_u \leq \sum_{u \in L_j \cap T_{v'}} x_u \leq 1/2$, and therefore the constraint at $v$ and $L_j$ is satisfied. 
Finally, if $\sum_{u \in P_v} x_u = 1$, there can be no special vertex below $v$ and therefore $\sum_{u \in L_j \cap T_v} x_u = 0$.  

\vspace{0.1in} 

\noindent {\bf Integral Solution:} 
We argue that any integral solution cannot save more than $(1+5/\alpha) 5 \alpha^5$ terminals. The following lemma is the key to our analysis. 

\begin{lemma} 
Any integral solution $\uset: \uset \cap \set{a(1), b(1)} =\emptyset$ saves at most $(1+5/\alpha) 5 \alpha^5$ terminals.
\label{intgap56:int1}
\end{lemma} 
\begin{appendixproof}{sec:newlp}{\subsection{Proof of Lemma \ref{intgap56:int1}}}
Consider the set $Q= \set{a'(1,j)}_{j=1}^{\alpha} \cup \set{b'(1,j)}_{j=1}^{\alpha}$, and a collection of paths from $\set{a(1),b(1)}$ to vertices in set $Q$.  
These paths are contained in the layers $L_1, \ldots, L_{\alpha+1}$, so the strategy $\uset$ induces a cut of size at most $\alpha+1$ on them. 
This implies that at most $\alpha+1$ vertices (out of $2\alpha$ vertices in $Q$) can be saved by $\uset$. 
Let $\tilde Q \subseteq Q$ denote the set of vertices that have not been saved by $\uset$. We remark that $|\tilde Q| \geq \alpha - 1$.
We write $\tilde Q = \tilde Q_a \cup \tilde Q_b$ where $\tilde Q_a$ contains the set of vertices $a'(1,j)$ that are not saved, and $\tilde Q_b = \tilde Q \setminus \tilde Q_a$.  
For each vertex in $\tilde Q_a$, at least $\alpha^4 -1$ of its children cannot be saved, so we have at least $(\alpha^4 -1) |\tilde Q_a| \geq \alpha^4 |\tilde Q_a| - \alpha$ unsaved terminals that are descendants of $\tilde Q_{a}$.
If $|\tilde Q_b| \leq 3$, we are immediately done: We have $|\tilde Q_a| \geq \alpha-4$, so $(\alpha^4-1)(\alpha-4) \geq \alpha^5 - 5\alpha^4$ unsaved terminals.  
 
Consider the set 
$$R = \paren{\bigcup_{j \in [\alpha],j' \in [\alpha^2]} \set{a''(1,j,j')}} \cup \paren{\bigcup_{j: b'(1,j) \in \tilde Q_b} \bigcup_{j' \in [\alpha^2]} \set{b''(1,j,j')}}$$  
This set satisfies $|R| = \alpha^3 + |\tilde Q_b| \alpha^2$, and the paths connecting vertices in $R$ to $\tilde Q_b \cup \set{a(1)}$ lie in layers $L_1,\ldots, L_{\alpha^3 + \alpha^2 + \alpha+ 1}$.
So the strategy $\uset$ induced on these paths disconnects at most $\alpha^3+ \alpha^2 + \alpha+1$ vertices.
Let $\tilde R \subseteq R$ contain the vertices in $R$ that are not saved by $\uset$, so we have $|\tilde R| \geq (|\tilde Q_b|-1)  \alpha^2 - \alpha - 1$, which is at least $(|\tilde Q_b| - 2) \alpha^2$.  
Each vertex in  $\tilde R$ has $\alpha^2$ children. We will have $(\alpha^2-1)$ unsaved terminals for each such vertex, resulting in a total of at least $(\alpha^2 -1) (|\tilde Q_b|-2)\alpha^2 \geq \alpha^4 |\tilde Q_b| - 4 \alpha^4$ terminals that are descendants of $b(1)$.   

In total, by summing the two cases, at least $(\alpha^4 |\tilde Q_a| - \alpha) + (\alpha^4 |\tilde Q_b| - 4 \alpha^4) \geq (|\tilde Q_a| + |\tilde Q_b|)\alpha^4 - 5\alpha^4 \geq \alpha^5 - 5\alpha^4$ terminals are not saved by $\uset$, thus concluding the proof.  
\end{appendixproof} 

\begin{lemma} 
Any integral solution $\uset: \uset \cap \set{a(2), b(2)} =\emptyset$ saves at most $(1+5/\alpha) 5 \alpha^5$ terminals.
\label{intgap56:int2}
\end{lemma}

Since nodes $a(1)$, $a(2)$, $b(1)$, $b(2)$ are in the first layer, it is only
possible to save one of them. Therefore, either Lemma \ref{intgap56:int1} or
Lemma \ref{intgap56:int2} apply, which concludes the analysis.

\section{Conclusion and Open Problems}

In this paper, we settled the integrality gap question for the standard LP relaxation. 
Our results ruled out the hope to use the canonical LP to obtain better approximation results.
While a recent paper settled the approximability status of the problem~\cite{2016arXiv160100271A}, the question whether an improvement over $(1-1/e)$ can be done via LP relaxation is of independent interest.   
We provide some evidences that Hartke's LP is a promising candidate for doing so. 
Another interesting question is to find a more general graph class that admits a constant approximation algorithm. 
We believe that this is possible for bounded treewidth graphs.

\bibliographystyle{splncs03}\bibliography{main-firefighters} 

\appendix 

\clearpage

\section{Omitted Figures}

\begin{figure}[b]
\centering
\includegraphics{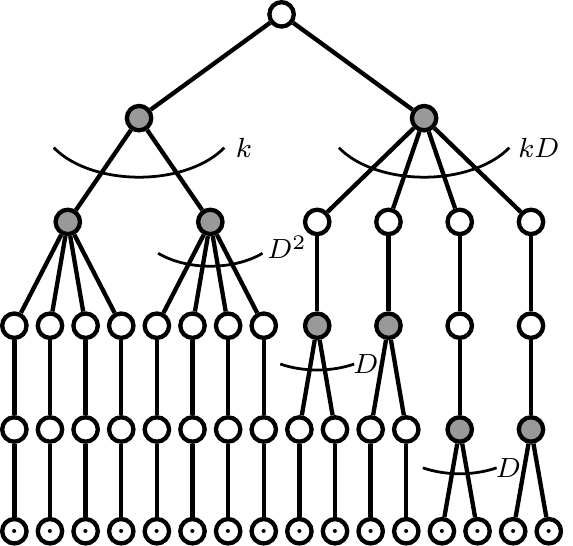}
\caption{Simplified example of the instance used to achieve integrality gap of $1-1/e$, when $k=2$ and $D=2$. The labels in the figure indicate, in general, the number of edges in that location, in terms of $k$ and $D$. Special vertices are colored gray.}
\label{11eintgap:gadget}
\end{figure}

\begin{figure}[b]
\centering
\includegraphics{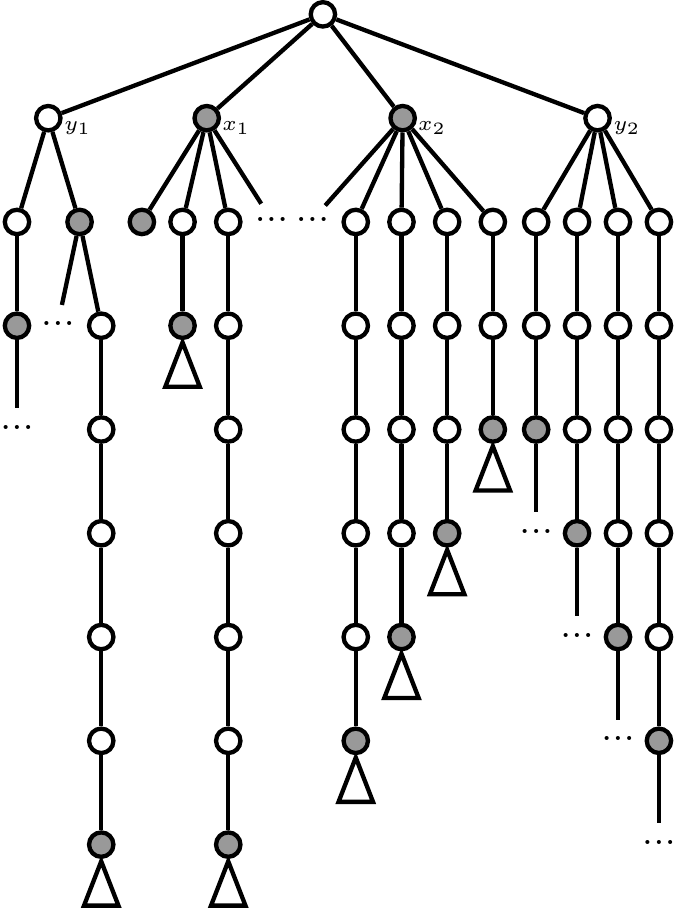}
\caption{Simplified example of the instance with low integrality gap for $1/2$-integral solutions. Special vertices are colored gray.}
\label{12intgap:instance}
\end{figure}

\end{document}